\documentclass[english,a4paper,12pt]{article}

\usepackage{a4wide}

\usepackage{babel}

\usepackage{graphicx}

\usepackage{enumerate}

\usepackage[T1]{fontenc}
\usepackage[utf8]{inputenc}

\usepackage[matrix,arrow,curve]{xy}
\usepackage{algorithm}
\usepackage{algorithmic}

\usepackage[tbtags]{amsmath}
\usepackage{amssymb,amsthm,stmaryrd}

\theoremstyle{plain}
\newtheorem{theorem}{Theorem}
\newtheorem*{theorem*}{Theorem}
\newtheorem{proposition}[theorem]{Proposition}
\newtheorem{corollary}[theorem]{Corollary}
\newtheorem*{corollary*}{Corollary}
\newtheorem{lemma}[theorem]{Lemma}

\theoremstyle{remark}
\newtheorem{remark}[theorem]{Remark}
\newtheorem{example}[theorem]{Example}

\theoremstyle{definition}
\newtheorem{definition}[theorem]{Definition}

\newcommand{\field}[1][]{\ensuremath{\mathbb{F}_{#1}}}
\newcommand{\bfield}{\field[]}
\newcommand{\Kfield}{\mathbb{K}}
\newcommand{\trace}{\operatorname{Tr}}
\newcommand{\Aut}[2][\bfield]{\operatorname{Aut}_{#1}(#2)}
\newcommand{\Der}[3][\bfield]{\operatorname{Der}_{#1}^{#3}(#2)}

\begin{document}

\title{Ideal codes over separable ring extensions.\thanks{Research partially supported by grant MTM2010-20940-C02-01 from the Ministerio de Ciencia e Innovaci\'{o}n of the Spanish Government and from FEDER, and by grant mP-TIC-14 (2014) from CEI-BioTic Granada. \newline
This paper was presented in part as a poster at ISSAC 2014.}
}

\author{Jos\'{e} G\'{o}mez-Torrecillas \and
        F. J. Lobillo \and
        Gabriel Navarro
}

\date{\today}

\maketitle

\begin{abstract}
This paper investigates the application of the theoretical algebraic notion of a separable ring extension, in the realm of cyclic convolutional codes or, more generally, ideal codes. We work under very mild conditions, that cover all previously known as well as new non trivial examples. It is proved that ideal codes are direct summands as left ideals of the underlying non-commutative algebra, in analogy with cyclic block codes. This implies, in particular, that they are generated by an idempotent element. Hence, by using a suitable separability element, we design an efficient algorithm for computing one of such idempotents. 
\end{abstract}

\section{Introduction}

Most of the codes used in engineering support a vector space structure (linear block codes) or become a direct summand of a free module over a polynomial ring (convolutional codes). In the linear case, the benefits are amplified if we also consider cyclicity, since the vector space is also endowed with an algebra structure and cyclic codes come to be ideals. Over convolutional codes, this notion requires something more sophisticated than a straightforward extension of the definition of cyclic block code. Piret, in \cite{Piret:1976}, shows that the classical notion of cyclicity does not produce non-block codes in the convolutional setting and he proposes to deal with skew polynomials (see also \cite{Roos:1979}). Unluckily, with the loss of commutativity on the ring, the working algebra presents theoretical problems that have hindered the study, and potencial practical applications, of cyclic convolutional codes. These codes are reconsidered by Gluesing-Luerssen and Schmale in \cite{GluesingSchmale:2004}, where it is proven that they are principal left ideals of an Ore extension of the form \(A[z;\sigma]\), where \(A = \bfield[x]/(x^n-1)\), $\bfield$ is a finite field, and $\sigma$ is an $\bfield$--automorphism of $A$. Note that $A$ is a semisimple commutative algebra over the finite field $\bfield$, since its characteristic is assumed to be coprime with $n$. In fact, there is always an idempotent generator of the code, which extends a well-known property of cyclic block codes. This result has been extended to any commutative finite semisimple algebra $A$ in \cite{Lopez/Szabo:2013}. When the algebra $A$ is not assumed to be commutative, there are also some positive results: In \cite{Estradaetal:2008} and \cite{Lopez/Szabo:2013} it is shown that convolutional codes which are left ideals of \(A[z;\sigma]\), where \(A\) is a semisimple group algebra, are also generated by idempotents, under suitable conditions on the automorphism $\sigma$. However, as observed in \cite{Lopez/Szabo:2013}, in general, it is not known if convolutional codes with this kind of additional algebraic structure are principal when $A$ is non-commutative. In this paper we aim to continue on this way in order to get a better understanding of $\sigma$-cyclic convolutional codes as well as extend the examples collected by the theory of ideal codes. 

We observe that a property shared by all the aforementioned cases is that the ring extension $\bfield [z] \subseteq A[z;\sigma]$ is separable in the sense of \cite{Hirata/Sugano:1966}. This is a generalization of the concept of separable algebra characterized by the existence of a separability element. With such an element at hand we derive in a constructive way that every convolutional code which is a left ideal of $A[z;\sigma]$ is generated by an idempotent, and we design an algorithm for computing it. Our method rests on the availability of a separability element of the extension $\bfield[z] \subseteq A[z;\sigma]$. We thus devote some efforts to describe explicitly these separability elements in a wide variety of situations which, in particular, include the cases studied in \cite{GluesingSchmale:2004,Estradaetal:2008,Lopez/Szabo:2013}. The use of the condition of separability allows us to work effectively with non-trivial extensions of the examples known until now, see for instance Example \ref{2x2matidemp}. This also illustrates how abstract mathematical results are applicable beyond the theoretical framework in which they were conceived.

The paper is structured as follows. In Section \ref{sec:2} we shall recall the notion of a separable extension of rings $S\subseteq R$ and the existing relation between $S$--direct summands and $R$--direct summands. Actually, given a left ideal $I$ of $R$ that it is an $S$--direct summand of $R$, we describe explicitly an idempotent $e$ of $R$ such that $I = Re$ (Corollary \ref{cor:SRsplit}). This description makes use of a separability element of the ring extension $S \subseteq R$. Hence, we prove that the separability element $p$ of a ring extension $B\subseteq A$ may be lifted to an extension of Ore polynomial rings, say $B[z;\sigma|_B]\subseteq A[z;\sigma]$, if $p$ is fixed under the action of the extension of $\sigma$ to $A \otimes_B A$ (Theorem \ref{separabletoOre}). Special attention is paid to the construction of such an invariant element $p$ in the case of separable algebras over a field (Theorem \ref{separableautos}). Section \ref{sec:idealcode} deals with ideal codes, so we consider a separable extension of the form $\bfield[z] \subseteq A[z;\sigma]$, where $A$ is a finite semisimple algebra over a finite field $\bfield$. Hence, ideal codes over $A[z;\sigma]$ are direct summands as left ideals (Proposition \ref{idealsumadirecta}). Theorems \ref{separabletoOre} and \ref{separableautos} supply a rich variety of separable ring extensions of the form $\bfield[z] \subseteq A[z;\sigma]$ and, henceforth, of ideal codes generated by idempotents. In particular, we cover group convolution codes (Proposition \ref{groupdirectsummand}), the case where $A$ is commutative (Proposition \ref{Aconm}), and a wide class of extensions where $A$ is a direct sum of matrix algebras (Theorem \ref{simplest}). 

Finally, in Section \ref{algoritmo}, we highlight the prominent role that the existence of this separability element plays in this theory and we design an algorithm for computing an idempotent generator of a given ideal code. All along the paper we complement the theoretical results with several examples showing its practical applicability. Elements in finite fields, except \(0\) and \(1\), are represented as powers of a primitive element.

A basic reference for the general theory of non-commutative rings and their (bi)modules is \cite{Anderson/Fuller:1992}, and for finite fields the reader is referred to \cite{Lidl/Niederreiter:1997}. A good monograph on the theory of separable algebras is \cite{DeMeyer/Ingraham:1971}, while for separable ring extensions we refer to \cite{Hirata/Sugano:1966}.

\section{Separability and Ore extensions}\label{sec:2}

Let $R, S, T$ be unital (possibly non-commutative) rings, $M$ an $R-S$--bimodule, and $N$ an $S-T$--bimodule. Then its tensor product $M \otimes_S N$ becomes an $R-T$--bimodule in the usual way. If there is no confusion, we will write $\sum_i m_i \otimes n_i \in M \otimes_S N$ instead of $\sum_i m_i \otimes_S n_i$. For a homomorphism of unital rings $\rho : S \to R$, we consider the canonical $S$--bimodule $R$ with actions $s r = \rho(s)r, rs = r\rho(s)$, for all $r \in R$ and $s \in S$. Clearly, $R$ becomes both an $R-S$--bimodule and an $S-R$--bimodule. 

\subsection{Separability and generating idempotents.}

 Let $M, N$ denote either left modules or bimodules over some ring. 
Let $f : M \to N$ and $g : N \to M$ be homomorphisms of (bi)modules such that $f \circ g = id_N$. Then $f$ is said to be a split epimorphism of (bi)modules and $g$ is said to be a split monomorphism of (bi)modules. For any submodule $N$ of a module $M$, we have the canonical exact sequence 
\begin{equation}\label{eq:sec1}
\xymatrix{0 \ar[r] & N \ar[r] & M \ar^-{\pi}[r] & M/N \ar[r] & 0,}
\end{equation}
where $\pi : M \to M/N$ denotes the canonical projection given by $\pi (m) = m + N$ for all $m \in M$. 
Then $N$ is a direct summand (as a module) of $M$ if and only if $\pi$ is a split epimorphism (see \cite[Proposition 5.2]{Anderson/Fuller:1992}). In this case, if $\iota : M/N \to M$ is the splitting monomorphism of $\pi$, that is, $\pi \circ \iota = id_{M/N}$, then a supplement of $N$ in $M$ is obtained as the image of $\iota$ (see \cite[Lemma 5.1]{Anderson/Fuller:1992}). We will represent this situation by the diagram
\[
\xymatrix{0 \ar[r] & N \ar[r] & M \ar^-{\pi}[r] & M/N \ar[r] \ar@/^/[l]^{\iota} & 0.}
\]

The following generalization of the classical notion of a separable algebra over a commutative ring is a key conceptual tool in this paper. We need this generalization because, even though that many of the rings we are interested in are extensions of a polynomial commutative ring, this subring will not be central. 

\begin{definition}[{\cite[Definition 2]{Hirata/Sugano:1966}}]
A homomorphism of rings $\rho : S \to R$ is said to be a \emph{separable ring extension} if the multiplication map $\mu : R \otimes_S R \to R$, that maps $r \otimes r'$ onto $rr'$, is a split epimorphism of $R$--bimodules. Equivalently, there exists an element $p = \sum_i a_i \otimes b_i \in R \otimes_S R$ (called a \emph{separability element}) such that $rp = pr$ for all $r \in R$ and $\mu(p) = 1$, that is, for all $r \in R$, 
\begin{equation}\label{separability1}
\sum_i r a_i \otimes b_i = \sum_i a_i \otimes b_i r,
\end{equation}
and 
\begin{equation}\label{separability2}
\mu({\textstyle\sum_i a_i \otimes b_i}) = \sum_i a_i b_i = 1.
\end{equation}
\end{definition}

\begin{remark}
If $\rho : S \to R$ is a separable ring extension, then $\rho$ is injective. Indeed, if $\rho (s) = 0$ for $s \in S$, then 
$s = s\mu(p) = \mu (sp) = \mu(0) = 0$. Thus, a separable ring extension is often denoted by $S \subseteq R$, even though that $\rho$ needs not to be an inclusion of sets.
\end{remark}

Let $S \subseteq R$ be a separable ring extension. If $N$ is an $R$--submodule of a left $R$--module $M$ such that $N$ is an $S$--direct summand, then $N$ is an $R$--direct summand (see \cite[Theorem 1.1(c), Definition 1, Proposition 2.6]{Hirata/Sugano:1966}). As a consequence, we deduce the following corollary, that becomes a key tool in this work. We include a proof because we will need the explicit expression of the homomorphism $\beta$ below. 

\begin{corollary}\label{cor:SRsplit}
Let $S \subseteq R$ be a separable ring extension with separability element $p = \sum_i a_i \otimes b_i \in R \otimes_S R$. Consider a left ideal $I$ of $R$ which is an $S$--direct summand of $R$ with $S$--split exact sequence 
\begin{equation}\label{canonicalshortsequence}
\xymatrix{0 \ar[r] & I \ar[r] & R \ar[r]^{\pi} & R/I \ar[r] \ar@/^/[l]^{\iota} & 0}.
\end{equation}
Then $I$ is an $R$--direct summand of $R$ with $R$--split exact sequence 
\begin{equation}\label{canonicalshortsequence2}
\xymatrix{0 \ar[r] & I \ar[r] & R \ar[r]^{\pi } & R/I \ar[r] \ar@/^/[l]^{\beta} & 0},
\end{equation}
where 
\begin{equation}\label{beta}
\beta (r + I) = \sum_ia_i\iota(b_ir+I) 
\end{equation}
 for every $r + I \in R/I$. Therefore, $I = Re$, where $e \in R$ is the idempotent $e = 1-f$, with 
 \[
 f = \beta(1+I) = \sum_ia_i\iota(b_i + I).
 \]
\end{corollary}
\begin{proof}
Since \eqref{canonicalshortsequence} is split, the sequence 
\begin{equation}\label{tensorshortsequence}
\xymatrix{
0 \ar[r] & R \otimes_S I \ar[r] & R \otimes_S R \ar[r] & R \otimes_S R/I \ar[r] \ar@/^/[l]^{R \otimes \iota} & 0 
}
\end{equation}
is also exact and splits as a sequence of left \(R\)--modules. For each left \(R\)--module \(M\), let \(\mu: R \otimes_S M \rightarrow M\) and \(\alpha: M \rightarrow R \otimes_S M\) be the left \(R\)--module maps defined by \(\mu(a \otimes x) = ax\) and \(\alpha(x) = \sum_i a_i \otimes b_i x\). As shown in the proof of \cite[Proposition 2.6]{Hirata/Sugano:1966}, \(\mu \circ \alpha = id_M \), hence it follows from the commutativity of the diagram
\[
\xymatrix{
0 \ar[r] & I \ar[r] \ar@/^/[d]^{\alpha} & R \ar[r] \ar@/^/[d]^{\alpha} & R/I \ar[r] \ar@/^/[d]^{\alpha} & 0 \\
0 \ar[r] & R \otimes_S I \ar[r] \ar@/^/[u]^{\mu} & R \otimes_S R \ar[r] \ar@/^/[u]^{\mu} & R \otimes_S R/I \ar[r] \ar@/^/[u]^{\mu} \ar@/^/[l]^{R \otimes \iota} & 0 
}
\] 
that \(\beta = \mu \circ (R \otimes \iota) \circ \alpha\) makes \eqref{canonicalshortsequence2} a split exact sequence of left \(R\)--modules, as desired. Finally, a straightforward computation shows that $f = \beta (1 + I)$ is idempotent and that it generates a left ideal $J$ of $R$ such that $R = I \oplus J$. On the other hand, $e = 1-f \in I$, since $f + I = \beta (1+I) + I = 1 +I$. Hence, $I = Re$. 
\end{proof}

\subsection{Separable Ore extensions.}
Our next goal is to extend separability to Ore extensions. Let \(\sigma\) be an endomorphism of a ring \(A\). A (right) \(\sigma\)--derivation is an additive map \(\delta: A \rightarrow A\) such that \(\delta(ab) = \delta(a)\sigma(b) + a \delta(b)\) for all \(a,b\in A\). Given \(\sigma\), the set of all \(\sigma\)--derivations is denoted by \(\Der[]{A}{\sigma}\). Let $B$ be a subring of $A$ such that $\sigma (B) \subseteq B$ and $\delta (B) \subseteq B$ for some $\delta \in \Der[]{A}{\sigma}$. Even though that \(\sigma\) and \(\delta\) need not to be \(B\)--bimodule maps, it is possible to extend them to maps \(\sigma^\otimes, \delta^\otimes : A \otimes_B A \to A \otimes_B A\) as the following lemma shows.

\begin{lemma}\label{sigmadeltatensor}
Let \(B \subseteq A\) be a ring extension, $\sigma$ an endomorphism of $A$, and $\delta \in \Der[]{A}{\sigma}$. If \(\sigma(B) \subseteq B\) and \(\delta(B) \subseteq B\), then the maps
\[
\begin{split}
\sigma^\otimes: A \otimes_B A &\longrightarrow A \otimes_B A \\
a \otimes b &\longmapsto \sigma(a) \otimes \sigma(b) \\
\delta^\otimes: A \otimes_B A &\longrightarrow A \otimes_B A \\
a \otimes b &\longmapsto \delta(a) \otimes \sigma(b) + a \otimes \delta(b)
\end{split}
\]
are well defined.
\end{lemma}

\begin{proof}
In order to show that \(\sigma^\otimes\) is well defined it is enough to check that \(\sigma^\otimes(as \otimes b) = \sigma^\otimes(a \otimes sb)\) for all \(s \in B\) and all \(a,b \in A\):
\[
\begin{split}
\sigma^\otimes(as \otimes b) &= \sigma(as) \otimes \sigma(b) \\
&= \sigma(a)\sigma(s) \otimes \sigma(b) \\
&= \sigma(a) \otimes \sigma(s) \sigma(b) \\
&= \sigma(a) \otimes \sigma(sb) = \sigma^\otimes(a \otimes sb),
\end{split}
\]
where we have used that \(\sigma\) is a homomorphism of rings and \(\sigma(s) \in B\). Analogously,
\[
\begin{split}
\delta^\otimes(as \otimes b) &= \delta(as) \otimes \sigma(b) + as \otimes \delta(b) \\
&= (\delta(a)\sigma(s) + a \delta(s)) \otimes \sigma(b) + as \otimes \delta(b) \\
&= \delta(a)\sigma(s) \otimes \sigma(b) + a\delta(s) \otimes \sigma(b) + as \otimes \delta(b) \\
&= \delta(a) \otimes \sigma(s)\sigma(b) + a \otimes \delta(s)\sigma(b) + a \otimes s\delta(b) \\
&= \delta(a) \otimes \sigma(sb) + a \otimes (\delta(s)\sigma(b) + s\delta(b)) \\
&= \delta(a) \otimes \sigma(sb) + a \otimes \delta(sb) = \delta^\otimes(a \otimes sb),
\end{split}
\]
where it is used that \(\delta\) is a \(\sigma\)--derivation and \(\delta(B) \subseteq B\). Hence \(\delta^\otimes\) is also well defined.
\end{proof}

\begin{remark}\label{sigmaseparable}
It is not hard to check that if $p \in A \otimes_B A$ is a separability element of $B \subseteq A$, and $\sigma$ is an automorphism of $A$ such that $\sigma (B) \subseteq B$, then $\sigma^\otimes (p)$ is a separability element of $B \subseteq A$. 
\end{remark}

Recall that the Ore extension \(A[z;\sigma,\delta]\) of $A$, where $\sigma$ is a ring endomorphism of $A$ and $\delta \in \Der[]{A}{\sigma}$, is the free right \(A\)--module with basis the powers of \(z\), and multiplication defined by the rule
\[
a z = z \sigma(a) + \delta(a), \text{ for all \(a \in A\)}.
\]
Ore extensions are also known as skew polynomial rings.
With this product, $A[z;\sigma,\delta]$ becomes a (typically non-commutative) ring (see, e.g. \cite[Section 1.2]{McConnell/Robson:2001})
whose elements are polynomials in \(z\) with coefficients on the right, and \(A \subseteq A[z;\sigma,\delta]\) is the subring of polynomials of degree \(0\). In order to prove the main result of this section we need to introduce some notation. Let \(B \subseteq A\) such that \(\sigma(B) \cup \delta(B) \subseteq B\), and let us denote \(R= A[z;\sigma,\delta]\) and \(S = B[z;\sigma_{|B},\delta_{|B}]\). Let \(\varphi\) be the morphism of \(A\)--bimodules defined as the composition of the canonical morphisms
\[
\varphi: A \otimes_B A \to R \otimes_B R \to R \otimes_{S} R.
\]

\begin{theorem}\label{separabletoOre}
Let \(B \subseteq A\) be a separable ring extension with separability element \(p \in A \otimes_{B} A\). Let $\sigma$ be an endomorphism of $A$, and $\delta \in \Der[]{A}{\sigma}$ such that \(\sigma(B) \subseteq B\) and \(\delta(B) \subseteq B\). If \(\sigma^\otimes(p) = p\) and \(\delta^\otimes(p) = 0\), then \(B[z;\sigma_{|B},\delta_{|B}] \subseteq A[z;\sigma,\delta]\) is a separable ring extension with separability element \(\overline{p} = \varphi(p)\).
\end{theorem}

\begin{proof}
Let $S = B[z;\sigma_{|B},\delta_{|B}]$, $R = A[z;\sigma,\delta]$. 
If \(p = \sum_{i} a_i \otimes_B b_i\), then \(\overline{p} = \varphi(p) = \sum_i a_i \otimes_{S} b_i\). For all \(a \in A\),
\[
a \overline{p} = a \varphi(p) = \varphi(ap) = \varphi(pa) = \varphi(p)a = \overline{p} a,
\]
so it remains to prove \eqref{separability1} for \(z\): 
\begin{displaymath}
\begin{split}
\overline{p} z &= \sum_{i} a_i \otimes_{S} b_i z \\
&= \sum_{i} a_i \otimes_{S} (z \sigma(b_i) + \delta(b_i)) \\
&= \sum_{i} a_i z \otimes_{S} \sigma(b_i) + \sum_i a_i \otimes_{S} \delta(b_i) \\
&= \sum_{i} (z \sigma(a_i) + \delta(a_i)) \otimes_{S} \sigma(b_i) + \sum_i a_i \otimes_{S}\delta(b_i) \\
&= z \sum_i \sigma(a_i) \otimes_{S} \sigma(b_i) + \sum_i \left( \delta(a_i) \otimes_{S} \sigma(b_i) + a_i \otimes_{S} \delta(b_i)\right) \\
&= z \varphi\left(\sum_i \sigma(a_i) \otimes_B \sigma(b_i)\right) + \varphi\left(\sum_i \left( \delta(a_i) \otimes_{B} \sigma(b_i) + a_i \otimes_B \delta(b_i)\right) \right) \\
&= z \varphi(\sigma^\otimes(p)) + \varphi(\delta^\otimes(p)) \\
&= z \varphi(p) = z \overline{p}
\end{split}
\end{displaymath} 
as desired. On the other hand,
\[
\mu (\overline{p}) = \sum_i a_i b_i = 1
\]
and the proof is completed.
\end{proof}

We close this section with some fundamental examples of separable Ore extensions of interest in the rest of the paper. 

\begin{example}\label{separablefinitefield}
Let \(\bfield = \field[q] \subseteq \field[q^t]\) be a finite field extension. Let $\sigma = \tau^h$ be an $\bfield$--automorphism, where $\tau$ denotes the Frobenius automorphism of the extension. As a consequence of Theorem \ref{separabletoOre}, the ring extension $\bfield [z] \subseteq \field[q^t][z;\sigma]$ is separable. To see this, we will exhibit a separability element $p \in \field[q^t] \otimes_{\bfield} \field[q^t]$ of the extension $\bfield \subseteq \field[q^t]$ such that $\sigma^\otimes(p) = p$. We follow \cite{Lidl/Niederreiter:1997} for basic facts concerning finite fields, in particular we follow the notation and properties about the trace function.
It is well known that a separability element can be obtained from dual bases. The dual basis of a normal basis is also normal, hence let \(\{a, a^q \dots, a^{q^{t-1}}\},\{b, b^q \dots, b^{q^{t-1}}\}\) be normal dual bases. We are going to check, for convenience of the reader, that \(p = \sum_i a^{q^i} \otimes b^{q^i} \in \field[q^t] \otimes_{\bfield} \field[q^t]\) is a separability element. Dual bases are characterized by the equalities \(\alpha = \sum_i \trace_{\field[q^t]/\bfield}(b^{q^i} \alpha) a^{q^i} = \sum_i \trace_{\field[q^t]/\bfield}(a^{q^i} \alpha) b^{q^i}\) for all $\alpha \in \field[q^t]$. Hence,
\begin{displaymath}
\begin{split}
\alpha p &= \sum_i \alpha a^{q^i} \otimes b^{q^i} \\
&= \sum_i \sum_j \trace_{\field[q^t]/\bfield}(b^{q^j} \alpha a^{q^i}) a^{q^j} \otimes b^{q^i} \\
&= \sum_j a^{q^j} \otimes \sum_i \trace_{\field[q^t]/\bfield}(a^{q^i} b^{q^j} \alpha) b^{q^i} \\
&= \sum_j a^{q^j} \otimes b^{q^j} \alpha \\
&= p \alpha,
\end{split}
\end{displaymath}
and \eqref{separability1} is satisfied. Moreover \(\sum_i a^{q^i} b^{q^i} = \sum_i (ab)^{q^i} = \trace_{\field[q^t]/\bfield}(ab) = 1\) by duality, and thus \eqref{separability2} also holds. Since \(\sigma(x) = x^{q^h}\) for all \(x \in \field[q^t]\), we get that \(\sigma(a^{q^i}) = a^{q^{i+h \pmod{t}}}\), and similarly for $b^{q^i}$. Therefore, \(\sigma\) induces the same permutation on \(\{a, a^q \dots, a^{q^{t-1}}\}\) and \(\{b, b^q \dots, b^{q^{t-1}}\}\), which clearly implies that $\sigma^{\otimes}(p) = p$. Therefore, a separability element for $\bfield[z] \subseteq \field[q^t][z;\sigma]$ is
\[
\overline{p} = \sum_i a^{q^i} \otimes_{\bfield[z]} b^{q^i}\]
\end{example}

\begin{example}\label{explmat}
Matrix rings give well known examples of separable ring extensions. Let \(A = \mathcal{M}_n(B)\) be the $n \times n$ matrix ring with entries in a given ring $B$, and $\sigma : A \to A$ be an automorphism such that $\sigma (B) \subseteq B$. Consider \(\{E_{ij} ~|~ 1 \leq i,j \leq n \} \) the matrix units, i.e. \(E_{ij}\) is the matrix with \(1\) in row \(i\) column \(j\) and \(0\) elsewhere. From the relations of the products of matrix units, it follows easily that, for all \(1 \leq j \leq n\), \(\sum_{i=1}^{n} E_{ij} \otimes E_{ji}\) are separability elements of the extension $B \subseteq A$. Although they do not need to be invariant under \(\sigma^\otimes\), it is possible to cover many cases using them, as the following example shows (see also Theorem \ref{simplest}). 
\end{example}

\begin{example}\label{2x2mat}
Let $A=\mathcal{M}_2(\field[8])$ be the ring of $2\times2$ matrices over the field $\field[8]=\field[2][\alpha]/(\alpha^3+\alpha+1)$. In what follows, we write the elements of $\field[8] \backslash \{0,1\}$ as powers of the primitive element $\alpha$, and not as polynomials. This convention is used in all applicable examples. Let $\sigma:A\to A$ be the automorphism given by $\sigma(X)=UXU^{-1}$, where
\[
U=\begin{pmatrix}
\alpha^{4} & 1 \\
1 & \alpha
\end{pmatrix}.
\]
The reader may check that the order of $\sigma$ is \(3\). Let 
\[
 p =\begin{pmatrix}
1 & 0 \\
0 & 0
\end{pmatrix} \otimes \begin{pmatrix}
1 & 0 \\
0 & 0
\end{pmatrix}+
\begin{pmatrix}
0 & 0 \\
1 & 0
\end{pmatrix} \otimes
\begin{pmatrix}
0 & 1 \\
0 & 0
\end{pmatrix}
\]
be a separability element of the extension $\field[8]\subseteq \mathcal{M}_2(\field[8])$ as explained in Example \ref{explmat}. Hence, since $(|\sigma|,\mathrm{char}(\field[8]))=1$, $ \overline{p} = |\sigma|^{-1} ( p +(\sigma\otimes\sigma)( p )+(\sigma^2\otimes\sigma^2)( p ))$ is a separability element of the extension $\field[8][z]\subseteq \mathcal{M}_2(\field[8])[z;\sigma]$. Explicitly,
\[
\begin{split}
\overline{p} &= \begin{pmatrix}
1 & 0 \\
0 & 0
\end{pmatrix} \otimes 
\begin{pmatrix}
1 & 0 \\
0 & 0
\end{pmatrix}+
\begin{pmatrix}
0 & 0 \\
1 & 0
\end{pmatrix} \otimes
\begin{pmatrix}
0 & 1 \\
0 & 0
\end{pmatrix}
+
\begin{pmatrix}
\alpha & 1 \\
\alpha^4 & \alpha^3
\end{pmatrix} \otimes 
\begin{pmatrix}
\alpha & 1 \\
\alpha^4 & \alpha^3
\end{pmatrix} \\ 
&\quad +
\begin{pmatrix}
\alpha^4 & \alpha^3 \\
\alpha^5 & \alpha^4
\end{pmatrix}
 \otimes
\begin{pmatrix}
1 & \alpha^4 \\
\alpha^3 & 1
\end{pmatrix}
+
\begin{pmatrix}
\alpha & \alpha^4 \\
1 & \alpha^3
\end{pmatrix} \otimes \begin{pmatrix}
\alpha & \alpha^4 \\
1 & \alpha^3
\end{pmatrix} \\ 
&\quad + 
\begin{pmatrix}
1 & \alpha^3 \\
\alpha^4 & 1
\end{pmatrix} \otimes
\begin{pmatrix}
\alpha^4 & \alpha^5 \\
\alpha^3 & \alpha^4
\end{pmatrix}.
\end{split}
\]
\end{example}

\subsection{Separable automorphisms.}\label{sec:separableautos}

Let $A$ be a separable algebra over a field $\Kfield$. Then $A$ is a finite dimensional semisimple $\Kfield$--algebra \cite[Example I, page 40]{DeMeyer/Ingraham:1971}. Consider the decomposition $1 = e_1 + \cdots + e_{n}$, where $e_1, \dots, e_{n}$ are (different) central idempotents of $A$ such that $Ae_i$ is a simple algebra for all $i= 1, \dots, n$. We call $\{ e_1, \dots, e_n \}$ a complete set of central idempotents for $A$. We have a block decomposition $A = Ae_1\oplus \cdots \oplus A e_{n}$ of $A$ into simple algebras which, in fact, are separable (as they are factor algebras of $A$, see \cite[Proposition 1.1]{DeMeyer/Ingraham:1971}). Actually, any set of separability elements of the algebras $Ae_i$ may be lifted to a separability element of $A$. This is a consequence of a more general result, given in Lemma \ref{sepidempt}. Recall that if $e$ is a central idempotent of a ring $A$, then the projection $A \to Ae$ that maps $a \in A$ onto $ae$ is a homomorphism of rings. 

\begin{lemma}\label{sepidempt}
Let $\rho : B \to A$ be a ring homomorphism, and assume that $1 = e_1 + e_2$, where $e_1, e_2$ are nontrivial central idempotents of $A$. If $\xymatrix{B \ar^{\rho}[r] & A \ar[r] & Ae_i}$ is a separable ring extension with separability element $p_i$ for $i = 1, 2$, then $\rho : B \to A$ is a separable ring extension with separability element $p = p_1 + p_2$. 
\end{lemma}
\begin{proof}
We have an $A$--bimodule decomposition $A = Ae_1 \oplus Ae_2$ which, obviously, is also a direct sum of $B$--bimodules. Therefore, since the tensor product $\otimes_B$ preserves direct sums, we may consider $Ae_i \otimes_B Ae_i$ as an $A$--subbimodule of $A \otimes_B A$ for $i=1,2$. Thus, the sum $p = p_1 + p_2$ makes sense in $A \otimes_B A$. 
Observe that, for $i = 1, 2$, the $B$--bimodule structure induced on $Ae_i$ by the ring homomorphism $\xymatrix{B \ar^{\rho}[r] & A \ar[r] & Ae_i}$ coincides with the given as an $B$--subbimodule of $A$. Now, for $r \in A$, 
\[
rp = rp_1 + rp_2 = re_1p_1 + re_2p_2 = p_1re_1+ p_2re_2 = p_1r + p_2 r = pr,
\]
and
\[
\mu (p) = \mu (p_1) + \mu (p_2) = e_1 + e_2 = 1,
\]
as desired.
\end{proof}

\begin{remark}\label{sepofproducts}
It follows from Lemma \ref{sepidempt} that if \(B \subseteq A_1\) and \(B \subseteq A_2\) are separable ring extensions with separability elements \(p_1 = \sum_i f_i \otimes g_i\) and \(p_2 = \sum_j k_j \otimes l_j\), respectively, then \(B \subseteq A_1 \times A_2\) is a separable extension with separability element 
\[
p = \sum_{i} (f_i,0) \otimes (g_i,0) + \sum_j (0,k_j) \otimes (0,l_j),
\]
where \(B\) is identified with its image in \(A_1 \times A_2\) via diagonal inclusion.
\end{remark}

Let $\sigma$ be a $\Kfield$--automorphism of our separable algebra $A$. It is easily checked that $\{ \sigma(e_1), \dots, \sigma(e_{n}) \}$ is a set of central idempotents of $A$ such that $1 = \sum_{i=1}^{n} \sigma(e_i)$, and that the restriction of $\sigma$ to each $Ae_i$ gives an algebra isomorphism $\sigma_i : Ae_i \to A\sigma(e_i)$. Therefore, the set $\{\sigma(e_1), \dots, \sigma(e_{n}) \}$ must be equal to $\{ e_1, \dots, e_{n} \}$ and $\sigma$ induces a permutation $\overline{\sigma}$ on $\{1, \dots, n \}$ such that $\sigma(e_i) = e_{\overline{\sigma}(i)}$, for all $i = 1, \dots, n$. 

Let $\{ 1, \dots, n \} = \bigcup_{j=1}^t Z_j$ be the partition of $\{1, \dots, n \}$ into orbits under the action of $\overline{\sigma}$. Then $A = \oplus_{j = 1}^t A^{(Z_j)}$, where $A^{(Z_j)} = \oplus_{i \in Z_j} Ae_i$. Moreover, $\sigma (A^{(Z_j)}) \subseteq A^{(Z_j)}$ for all $j = 1, \dots, t$, therefore inducing by restriction an automorphism $\sigma^{(Z_j)} $ of $A^{(Z_j)}$. 
The following lemma is a direct consequence of the previous discussion and Lemma \ref{sepidempt}.

\begin{lemma}\label{partition}
 If $p_j \in A^{(Z_j)} \otimes A^{(Z_j)}$ is a separability element such that ${\sigma^{(Z_j)}}^\otimes (p_j) = p_j$ for all $j = 1, \dots, t$, then $p = \sum_{j=1}^t p_j$ is a separability element of $A$ such that $\sigma^\otimes (p) = p$. 
\end{lemma}

Observe that, in Lemma \ref{partition}, each block $A^{(Z_j)}$ is a direct sum of finitely many isomorphic simple separable algebras, and that the corresponding permutation $\overline{\sigma^{(Z_j)}}$ on $Z_j$ is cyclic. We thus study this case separately. 

\begin{proposition}\label{separablecyclic}
Let $\sigma$ be an automorphism of a separable algebra $B$ over a field $\Kfield$, with block decomposition into simple algebras $B = \bigoplus_{i = 1}^m B_i$. Assume that $\overline{\sigma}$ permutes cyclically $\{1, \dots, m \}$ in the natural order. Let $\sigma_i : B_i \to B_{i+1}$ be the isomorphism induced by restriction of $\sigma$ to $B_i$ for every $i = 1, \dots, m$, assuming that $B_{m+1} = B_1$. If $p_1 \in B_1 \otimes B_1$ is a separability element such that $(\sigma_m \circ \cdots \circ \sigma_2 \circ \sigma_1)^{\otimes} (p_1) = p_1$, then 
\begin{equation}
p = p_1 + \sum_{i =1}^{m-1} \sigma_i^{\otimes} \circ \cdots \circ \sigma_1^{\otimes} (p_1) 
\end{equation} 
is a separability element of $B$ such that $\sigma^{\otimes}(p) = p$.
\end{proposition}
\begin{proof}
First, a word on notation: on the analogy of Lemma \ref{sigmadeltatensor}, $\sigma_i^\otimes$ denotes $\sigma_i \otimes \sigma_i$, for all $i = 1, \dots, n$. Observe that $\sigma_i^{\otimes} \circ \cdots \circ \sigma_1^{\otimes} =(\sigma_i \circ \cdots \circ \sigma_1)^{\otimes}$. Now, 
\[
\begin{split}
\sigma^{\otimes}(p) &= \sigma^{\otimes}(p_1) + \sum_{i=1}^{m-1}\sigma^{\otimes} \circ \sigma_i^{\otimes} \circ \cdots \circ \sigma_1^{\otimes}(p_1) \\ 
&= \sigma_1^{\otimes}(p_1) + \sum_{i=1}^{m-1}\sigma_{i+1}^{\otimes} \circ \sigma_i^{\otimes} \circ \cdots \circ \sigma_1^{\otimes}(p_1) \\
&= \sigma_1^{\otimes}(p_1) + \sum_{j=2}^{m-1}\sigma_{j}^{\otimes} \circ \sigma_{j-1}^{\otimes} \circ \cdots \circ \sigma_1^{\otimes}(p_1) + \sigma_m^\otimes \circ \dots \circ \sigma_1^\otimes (p_1) \\
&= \sum_{i =1}^{m-1} \sigma_i^{\otimes} \circ \cdots \circ \sigma_1^{\otimes} (p_1) + p_1 \\
&= p,
\end{split}
\]
since $\sigma_m^{\otimes} \circ \cdots \circ \sigma_1^{\otimes} (p) = (\sigma_m \circ \cdots \circ \sigma_1)^{\otimes}(p)$. On the other hand, $\sigma_i^{\otimes} \circ \cdots \circ \sigma_1^{\otimes} (p)$ is a separability element of the algebra $B_{i+1}$ for every $i = 1, \dots, m-1$. By Lemma \ref{sepidempt}, $p$ becomes a separability element for $B$, which completes the proof.
\end{proof}

We thus deduce Theorem \ref{separableautos} below, as a consequence of Lemma \ref{partition} and Proposition \ref{separablecyclic}. 

\begin{theorem}\label{separableautos}
Let $\sigma$ be an automorphism of a separable $\Kfield$--algebra $A$, and
write each orbit $Z_j = \{j_1, \dots, j_{m_j} \}$, for some $m_1, \dots, m_t > 0$, in such a way that $\overline{\sigma}$ acts as the cyclic permutation $(j_1, \dots, j_{m_j})$ on $Z_j$ for $j=1, \dots, t$. Assume that, for every $j =1, \dots, t$, there exists a separability element $p_j \in Ae_{j_1} \otimes Ae_{j_1}$ such that $(\sigma_{j_{m_j}} \circ \cdots \circ \sigma_{j_2} \circ \sigma_{j_1})^{\otimes} (p_j) = p_j$. Then
\[
p = \sum_{j=1}^t p_j + \sum_{j=1}^t \sum_{i=1}^{m_{j}-1} \sigma_{j_i}^{\otimes} \circ \cdots \circ \sigma_{j_1}^{\otimes} (p_j)
\]
is a separability element of $A$ such that $\sigma^{\otimes}(p) = p$. 
\end{theorem}

\section{Ideal codes generated by idempotents.}\label{sec:idealcode}

Let $\bfield = \field[q]$ be the finite field with $q$ elements, and consider a finite semisimple $\bfield$--algebra $A$. Let $\sigma$ be an $\bfield$--automorphism of $A$, and $\delta$ an $\bfield$--linear $\sigma$--derivation. Then the commutative polynomial ring $\bfield [z]$ is a subring of the Ore extension $R=A[z;\sigma,\delta]$ in the obvious way (that is, $\bfield [z]$ is the $\bfield$--subalgebra of $R$ generated by $z$). Thus, every left $R$--module becomes an $\bfield [z]$--module by restriction of scalars and, in particular, $R$ may be considered as an $\bfield [z]$-- module in this way. Every basis $B = \{v_0, \dots, v_{n-1}\}$ of $A$ as a vector space over $\bfield$ becomes also a basis of $R$ as an $\bfield [z]$--module. In fact, the map 
\[
\xymatrix@R=1pt{\mathfrak{p} : \bfield [z]^n \ar[r] & R \\
 (f_j(z))_{j=0}^{n-1} \ar@{|->}[r] & \sum_{j=0}^{n-1} f_j(z)v_j}
\]
is an isomorphism of left $\bfield [z]$--modules. Observe that if $f_j(z) = \sum_{i=0}^m z^ia_{ij}$, for some $a_{ij} \in \bfield$, then $\mathfrak{p} (\sum_{i=0}^m z^ia_{ij})_{j=0}^{n-1} = \sum_{j=0}^{n-1}(\sum_{i=0}^m z^ia_{ij})v_j = \sum_{i=0}^m z^i (\sum_{j=0}^{n-1} a_{ij}v_j)$.

\begin{definition}[\cite{Lopez/Szabo:2013}]
An \emph{ideal code} is an $\bfield [z]$--submodule direct summand $C$ of $\bfield [z]^n$ (that is, a convolutional code) such that $\mathfrak{p} (C)$ is a left ideal of $R$. 
\end{definition}

\begin{remark}
Once fixed the isomorphism $\mathfrak{p}$ (that is, the basis $B$), an ideal code is equivalently given by a left ideal $I$ of $R$ such that $\mathfrak{p}^{-1}(I)$ is an $\bfield [z]$--direct summand of $\bfield [z]^n$. We may thus understand that an ideal code is just a left ideal $I$ of $R$ which is an $\bfield [z]$--direct summand of $R$. We get from Corollary \ref{cor:SRsplit} the following proposition.
\end{remark}

\begin{proposition}\label{idealsumadirecta}
If $\bfield [z] \subseteq A[z;\sigma,\delta]$ is a separable ring extension, then every ideal code is a direct summand of $A[z;\sigma,\delta]$ as a left ideal and, hence, it is generated by an idempotent of $A[z;\sigma,\delta]$. 
\end{proposition}

Our next goal is to show that Proposition \ref{idealsumadirecta} can be applied to many examples, both previously considered by other authors, as well as introduced here for the first time. Let us start with group convolutional codes. 

\begin{definition}\cite{Estradaetal:2008}
Let $\bfield\mathcal{G}$ denote the group algebra of a finite group $\mathcal{G}$. Let \(\sigma \in \Aut{\bfield\mathcal{G}}\) and \(\delta \in \Der{\bfield\mathcal{G}}{\sigma}\). A group convolutional code is an ideal code in \(\bfield\mathcal{G}[z;\sigma,\delta]\).
\end{definition}

We get, as a consequence of our general results, the following proposition. An explicit separability element is provided in its proof, which will allow to apply Algorithm \ref{idem} for computing an idempotent generator of each group convolutional code. 

\begin{proposition}\cite[Proposition 3.6]{Lopez/Szabo:2013}\label{groupdirectsummand}
Let \(\mathcal{G}\) be a finite group such that \((|\mathcal{G}|,\operatorname{char}\bfield) = 1\), let \(\sigma \in \Aut{\bfield \mathcal{G}}\) and \(\delta \in \Der{\bfield\mathcal{G}}{\sigma}\) such that \(\sigma(\mathcal{G}) = \mathcal{G}\) and \(\delta(\mathcal{G}) = 0\). Then each group convolutional code is a direct summand of \(R = \bfield\mathcal{G}[z;\sigma,\delta]\) as a left ideal over \(R\), and, hence, generated by an idempotent of $R$. 
\end{proposition}
\begin{proof}
In view of Proposition \ref{idealsumadirecta}, we only need to prove that the ring extension $\bfield [z] \subseteq A[z;\sigma,\delta]$ is separable with $A = \bfield \mathcal{G}$. It is easily checked that 
\(p = |G|^{-1} \sum_{g \in \mathcal{G}} g \otimes g^{-1} \in \bfield\mathcal{G} \otimes_{\bfield} \bfield\mathcal{G}\) is a separability element for the extension \(\bfield \subseteq \bfield\mathcal{G}\) such that $\sigma^\otimes(p) = p$ and $\delta^\otimes(p) = 0$. By Theorem \ref{separabletoOre}, the extension $\bfield [z] \subseteq R$ is separable, with separability element
\begin{equation}\label{sepele1}
\overline{p} = |G|^{-1}\sum_{g \in \mathcal{G}}g \otimes_{\bfield[z]} g^{-1}
\end{equation}
which finishes the proof. 
\end{proof}

 The next type of examples arises when the finite semisimple algebra $A$ is assumed to be commutative (which, in particular, include the cyclic convolutional codes from \cite{GluesingSchmale:2004}). We should first do the following remark. 

\begin{remark}
In \cite[Theorem 3.5]{Lopez/Szabo:2013} it is claimed that if $A$ is a commutative finite semisimple $\bfield$--algebra, then every ideal code of $A[z;\sigma,\delta]$ is an $A[z;\sigma,\delta]$--direct summand. However, the proof of \cite[Lemma 3.3]{Lopez/Szabo:2013}, needed to derive \cite[Theorem 3.5]{Lopez/Szabo:2013}, has a gap at line 20 of page 962, as the following example shows: 
Consider $A = \bfield \times \bfield$, where $\bfield = \field[2]$ is the field with two elements. Let $\sigma : A \to A$ be the $\bfield$--algebra automorphism defined by $\sigma(a,b) = (b,a)$ for all $(a, b) \in A$, and consider the $\sigma$-derivation $\delta : A \to A$ defined by $\delta (a,b) = (a+b,0)$ for $(a,b) \in A$. Consider the ring $R = A[z;\sigma,\delta]$. The Ore polynomial $\alpha(z) = z(1,0) + (1,1) \in R$ generates a left ideal $I = R\alpha$. Let $\pi : R \to I$ be defined as follows. Every $h(z) \in R$ has a unique expression $h(z) = h_1(z) (1,0) + h_2(z)(0,1)$ for $h_1(z), h_2(z) \in \bfield [z]$. Define $\pi (h(z)) = (h_1(z) + zh_2(z))\alpha (z)$. Clearly, $\pi$ is a homomorphism of $\bfield [z]$--modules. Moreover, 
\[
h(z)\alpha(z) = (h_1(z) (1,0) + h_2(z)(0,1))\alpha (z) = h_2(z)\alpha(z),
\]
since $(1,0)\alpha(z) = 0$ and $(0,1)\alpha(z) = \alpha(z)$. Therefore,
\[
\begin{array}{lcl}
\pi(h(z)\alpha(z)) & = & \pi (h_2(z)\alpha(z)) \\
 &=& \pi (h_2(z)(z(1,0) + (1,1))) \\
 &=& \pi ((h_2(z)z + h_2(z))(1,0) + h_2(z)(0,1)) \\
 & = & (h_2(z)z + h_2(z) + zh_2(z))\alpha(z) \\
 & = & h_2(z) \alpha (z) \\
 & = & h(z)\alpha (z)
\end{array}
\]
According to the claim at the line 20 of the proof of \cite[Lemma 3.3]{Lopez/Szabo:2013}, one should have that $\pi (0,1) \in (0,1)A[z;\sigma,\delta]$. This would imply that $(0,1)\pi(0,1) = \pi(0,1)$. However, 
\[
\pi(0,1) = z\alpha(z) = z(z(1,0) + 1) = z^2(1,0) + z \neq 0,
\]
while
\[
(0,1)\pi(0,1) = (0,1)z\alpha(z) = (z(1,0) + (1,0))\alpha(z) = 0.
\]
Therefore, $\pi (0,1) \notin (0,1)A[z;\sigma,\delta]$. 
\end{remark}

The proof of \cite[Lemma 3.3]{Lopez/Szabo:2013}, and so that of \cite[Theorem 3.5]{Lopez/Szabo:2013}, is correct in the case $\delta = 0$. We also obtain this result as a consequence of our general methods. In addition, our proof provides a separability element that allows the computation of an idempotent generator of any ideal code in this setting, according to the algorithm described in Section \ref{algoritmo}.

\begin{proposition}\cite[Theorem 3.5]{Lopez/Szabo:2013}\label{Aconm}
Let $A$ be any finite semisimple commutative $\bfield$--algebra $A$, and $\sigma$ an $\bfield$--automorphism of $A$. Then every ideal code of $R=A[z;\sigma]$ is a direct summand left ideal of $R$ and, consequently, it is generated by an idempotent element of $R$.
\end{proposition}
\begin{proof}
Again by Proposition \ref{idealsumadirecta}, we need only to argue that $\bfield [z] \subseteq R$ is separable. Let $\{e_1, \dots, e_n\}$ be idempotents such that $A = \oplus_{i=1}^nAe_i$ is a decomposition of $A$ into simple blocks. Since $A$ is commutative, $Ae_i$ is a finite field extension of $\bfield$ for all $i = 1, \dots, n$. We follow the notation of Theorem \ref{separableautos}. Thus, for each $j =1, \dots, t$, we have an $\bfield$--automorphism $\sigma_{j_{m_j}} \circ \cdots \circ \sigma_{j_2} \circ \sigma_{j_1}$ of the finite field $Ae_{j_1}$. By Example \ref{separablefinitefield}, there is a separability element $p_j \in Ae_{j_1} \otimes Ae_{j_1}$ such that $(\sigma_{j_{m_j}} \circ \cdots \circ \sigma_{j_2} \circ \sigma_{j_1})^\otimes (p_j) = p_j$. Theorem \ref{separableautos} gives then a separability element $p \in A \otimes A$ such that $\sigma^{\otimes}(p) = p$, and Theorem \ref{separabletoOre} shows that $\bfield [z] \subseteq R$ is a separable ring extension, with separability element
\begin{equation}\label{sepele2}
\overline{p} = \sum_{j=1}^t \sum_k a_{jk} \otimes_{\bfield[z]} b_{jk} + \sum_{j=1}^t \sum_{i=1}^{m_{j}-1} \sum_k \sigma_{j_i} \circ \cdots \circ \sigma_{j_1}(a_{jk}) \otimes_{\bfield[z]} \sigma_{j_i} \circ \cdots \circ \sigma_{j_1} (b_{jk}),
\end{equation}
where $\{a_{jk}\}, \{b_{jk}\}$ denote dual normal bases of $Ae_{j1}$ over $\bfield$ for all $j = 1, \dots, t$. 
\end{proof}

It is known that any (possibly non-commutative) finite semisimple $\bfield$--algebra $A$ is separable. In fact, it is a direct sum of finitely many matrix rings with coefficients in (finite) field extensions of $\bfield$. If $B = M_n(\Kfield)$ is one of these simple blocks, with $\Kfield$ a finite field extension of $\bfield$, then we know that the ring extensions $\bfield \subseteq \Kfield$ and $\Kfield \subseteq M_n(\Kfield)$ are separable (see Examples \ref{separablefinitefield} and \ref{explmat}). By \cite[Proposition 2.5]{Hirata/Sugano:1966}, the extension $\bfield \subseteq M_n(\Kfield)$ is separable. In this way, our method can be applied to get idempotent generators for ideal codes built from Ore extensions of the form $A[z;\sigma]$, with $\sigma$ an $\bfield$--automorphism of $A$. Actually, Theorem \ref{separabletoOre}, in conjunction with Theorem \ref{separableautos}, shows that, in order to get a separability element for the extension $\bfield[z] \subseteq A[z;\sigma]$, it suffices to find separability elements of extensions of the form $\bfield \subseteq M_n(\Kfield)$ invariant under $\omega^\otimes$, where $\omega$ is some $\bfield$--automorphism of $M_n(\Kfield)$. Although we do not know if such an idempotent does exist for any $\omega$ (and we think this question deserves further investigation), there are neat situations where it is the case. One of the simplest is the following.

\begin{theorem}\label{simplest}
Let $\tau$ be an automorphism of the matrix algebra $B = \mathcal{M}_n(\bfield)$ such that $\tau^m = \operatorname{id}_B$ for some $m \geq 1$. Consider the algebra $A = \oplus_{i=1}^m B_i$, where $B_i = B$ for all $i = 1, \dots, m$. Let $\sigma: A \to A$ be defined by
\[
\sigma (b_1, b_2, \dots, b_m) = (\tau (b_m), \tau(b_1), \dots, \tau (b_{m-1})),
\]
for all $(b_1, b_2, \dots, b_m) \in A$. Then every ideal code in $A[z;\sigma]$ is a left ideal direct summand and, consequently, it is generated by an idempotent of $A[z;\sigma]$. 
\end{theorem}

\begin{proof}
Proposition \ref{idealsumadirecta} says that we only need to check that $\bfield[z] \subseteq A[z;\sigma]$ is a separable ring extension. Let $e_k$ be the central idempotent of $A$ with zeroes in all its components except the $k$--th, whose entry is $1$. The elements of \(B_k\) can therefore be represented as \(b e_k\) where \(b \in B = \mathcal{M}_n(\bfield)\). Let $p_1 \in B_1 \otimes B_1 = \mathcal{M}_n(\bfield) e_1 \otimes \mathcal{M}_n(\bfield)e_1$ be any separability element of the matrix algebra, e. g.
\begin{equation}\label{explicitchoice}
p_1 = \sum_{i=1}^n E_{ij}e_1 \otimes E_{ji}e_1,
\end{equation}
for some $j \in \{ 1, \dots, n\}$. Following the notation of Proposition \ref{separablecyclic}, we have that \(\sigma_i(be_i) = \tau(b) e_{i+1}\), hence
\[
(\sigma_m \circ \cdots \circ \sigma_2 \circ \sigma_1)^{\otimes} (p_1) = (\tau^m)^\otimes (p_1) = p_1.
\]
By Proposition \ref{separablecyclic}, 
\[
p = p_1 + \sum_{i=1}^{m-1} \sigma_i^{\otimes} \circ \cdots \circ \sigma_1^{\otimes} (p_1)
\]
is a separability element of $\bfield \subseteq A$ such that $\sigma^\otimes(p) = p$. Theorem \ref{separabletoOre} concludes the proof. In particular, for any explicit choice of \(p_1\) as in \eqref{explicitchoice}, 
\begin{equation}\label{sepele3}
\overline{p} = \sum_{i=1}^n E_{ij}e_1 \otimes_{\bfield[z]} E_{ji}e_1 + \sum_{k =1}^{m-1} \sum_{i=1}^n \tau^{k}(E_{ij})e_{k+1} \otimes_{\bfield[z]} \tau^k(E_{ji})e_{k+1}, 
\end{equation}
is a separability element of the ring extension $\bfield [z] \subseteq A[z;\sigma]$.
\end{proof}

\begin{example}
Let \(\bfield = \field[4] = \field[2](\alpha)\), \(B = \mathcal{M}_2(\bfield)\) and \(A = B \oplus B\). Let \(\tau \in \Aut{B}\) be the inner automorphism associated to \(U = \left(\begin{smallmatrix}\alpha^2 & \alpha^2 \\ \alpha & \alpha^2 \end{smallmatrix}\right)\), i.e. \(\tau(b) = UbU^{-1}\). Since \(U^2 = I_2\) it follows that \(\tau^2 = \operatorname{id}_B\). Hence \(\sigma : A \to A\), defined by \(\sigma(b_1,b_2) = (\tau(b_2),\tau(b_1))\), fits the hypothesis of Theorem \ref{simplest}. Therefore, the element
\[
\begin{split}
p &= \left( \left(\begin{smallmatrix} 1 & 0 \\ 0 & 0 \end{smallmatrix}\right), \left(\begin{smallmatrix} 0 & 0 \\ 0 & 0 \end{smallmatrix}\right)\right) \otimes \left( \left(\begin{smallmatrix} 1 & 0 \\ 0 & 0 \end{smallmatrix}\right), \left(\begin{smallmatrix} 0 & 0 \\ 0 & 0 \end{smallmatrix}\right)\right) + \left( \left(\begin{smallmatrix} 0 & 0 \\ 1 & 0 \end{smallmatrix}\right), \left(\begin{smallmatrix} 0 & 0 \\ 0 & 0 \end{smallmatrix}\right)\right) \otimes \left( \left(\begin{smallmatrix} 0 & 1 \\ 0 & 0 \end{smallmatrix}\right), \left(\begin{smallmatrix} 0 & 0 \\ 0 & 0 \end{smallmatrix}\right)\right) \\
&\quad + \left( \left(\begin{smallmatrix} 0 & 0 \\ 0 & 0 \end{smallmatrix}\right), \left(\begin{smallmatrix} \alpha^2 & \alpha^2 \\ \alpha & \alpha \end{smallmatrix}\right)\right) \otimes \left( \left(\begin{smallmatrix} 0 & 0 \\ 0 & 0 \end{smallmatrix}\right), \left(\begin{smallmatrix} \alpha^2 & \alpha^2 \\ \alpha & \alpha \end{smallmatrix}\right)\right) + \left( \left(\begin{smallmatrix} 0 & 0 \\ 0 & 0 \end{smallmatrix}\right), \left(\begin{smallmatrix} \alpha^2 & \alpha^2 \\ \alpha^2 & \alpha^2 \end{smallmatrix}\right)\right) \otimes \left( \left(\begin{smallmatrix} 0 & 0 \\ 0 & 0 \end{smallmatrix}\right), \left(\begin{smallmatrix} \alpha & \alpha^2 \\ 1 & \alpha \end{smallmatrix}\right)\right)
\end{split}
\]
is a separability element satisfying \(\sigma^\otimes(p) = p\). It has been obtained from \eqref{sepele3} with \(j = 1\) in \eqref{explicitchoice}. It follows from Theorem \ref{separabletoOre} that \(\overline{p} = \varphi(p)\) is a separability element for the extension \(\bfield[z] \subseteq A[z;\sigma]\).
\end{example}

\section{Computation of the idempotent generator.}\label{algoritmo}
In this section \(A\) is a finite semisimple algebra over a finite field $\bfield$. Fix a basis \(\{v_0, \dots, v_{n-1}\}\) of \(A\) as an \bfield{}--vector space. Let \(\sigma \in \Aut{A}\) and \(\delta \in \Der{A}{\sigma}\), and consider $R = A[z;\sigma,\delta]$ the corresponding Ore extension. Assume that $\bfield [z] \subseteq A[z;\sigma,\delta]$ is a separable ring extension, and that a separability element 
\(
\overline{p} = \sum_i a_i \otimes_{\bfield [z]} b_i \in R \otimes_{\bfield [z]} R
\)
is given. Theorem \ref{separabletoOre} provides a way to obtain such an element from a suitable separability element for the extension $\bfield \subseteq A$. This idea has been applied to several specific situations in Section \ref{sec:idealcode}. 

We know from Proposition \ref{idealsumadirecta} that every ideal code given by a left ideal $I$ of $R$ must be a direct summand of $R$ and, hence, it is generated, as a left ideal of $R$, by an idempotent. Our aim is to describe an algorithm that computes this idempotent from a set of generators $G = \{g_0, \dots, g_{t-1}\}$ of $I$ as a left ideal of $R$, whenever the separability element \(\overline{p}\) of the extension $\bfield[z] \subseteq A[z;\sigma,\delta]$ is available. A method for constructing such an element in a wide class of examples is provided in Subsection \ref{sec:separableautos}. The explicit formula of a separability element is given for group codes \eqref{sepele1}, ideal codes over any commutative finite semisimple algebra, henceforth including CCC codes \eqref{sepele2}, and for ideal codes over some non-commutative finite semisimple algebras \eqref{sepele3}. 

 Let \(I\) be a left ideal of \(R\). We have an exact sequence of left $R$--modules
\[
\xymatrix{R^t \ar[r]^{\cdot G} & R \ar[r]^{\pi} & R/I \ar[r] & 0},
\]
where $\cdot G$ is the homomorphism of left $R$--modules defined by right multiplication by the column matrix $G$. We know that $\{v_0, \dots, v_{n-1} \}$ becomes a basis of the $\bfield [z]$--module $R$, and we have the isomorphism of $\bfield [z]$--modules 
\[
\xymatrix@R=1pt{\mathfrak{p}: \bfield [z]^n \ar[r] & A[z;\sigma,\delta] \\ 
(f_0(z), \dots, f_{n-1}(z)) \ar@{|->}[r] & \sum_{j=0}^{n-1}f_j(z)v_j}
\]
Its inverse is
\[
\xymatrix@R=1pt{\mathfrak{v}: A[z;\sigma,\delta] \ar[r] &\bfield[z]^n \\
\sum_i z^i f_i \ar@{|->}[r] & (\textstyle\sum_i z^i f_{i,0}, \dots, \sum_i z^i f_{i,n-1}),}
\]
where, for all \(i\), \(f_i = f_{i,0} v_0 + \dots + f_{i,n-1} v_{n-1}\). 
Then \(I\) is generated as an \(\bfield[z]\)--module by \(\{v_i g_j ~|~ 0 \leq i \leq n-1, 0 \leq j \leq t-1 \}\). Hence a generator matrix for \(I\) is 

\begin{equation}\label{circulantmatrix}
M(G) = 
\begin{pmatrix}
\mathfrak{v}(v_0 g_0) \\
\vdots \\
\mathfrak{v}(v_{n-1} g_0) \\
\hline 
\vdots \\
\hline
\mathfrak{v}(v_0 g_{t-1}) \\
\vdots \\
\mathfrak{v}(v_{n-1} g_{t-1}) \\
\end{pmatrix},
\end{equation}
and we obtain a commutative diagram of homomorphisms of $\bfield[z]$--modules with exact rows
\[
\xymatrixcolsep{4pc}\xymatrixrowsep{3pc}\xymatrix{
R^t \ar[r]^{\cdot G} \ar@/^/[d]^{\mathfrak{v}}& R \ar[r]^{\pi} \ar@/^/[d]^{\mathfrak{v}} & R/I \ar[r] \ar@{=}[d]& 0 \\
\bfield[z]^{tn} \ar@/^/[u]^{\mathfrak{p}} \ar[r]^{\cdot M(G)} & \bfield[z]^n \ar@/^/[u]^{\mathfrak{p}} \ar[r]^{\pi \mathfrak{p}} & R/I \ar[r] & 0 }
\]

The left ideal $I$ is an ideal code if and only if it is a $\bfield[z]$--direct summand of $R$, equivalently, if and only if the Smith canonical form $H$ of $M(G)$ is basic, that is, $H$ is the matrix of size $tn \times n$ given by
\[ H=\begin{pmatrix} I_k & 0 \\ 0& 0 \end{pmatrix},
\] 
where $k$ is the dimension of the code, and $I_k$ is the identity matrix of order $k$. Let $P$ and $Q$ be invertible matrices with coefficients in $\bfield[z]$ and suitable sizes such that $PM(G)Q = H$. Let also $V=\left(\begin{smallmatrix} 0 \\ I_{n-k} \end{smallmatrix}\right)$, and $V^{T}$ the transpose of $V$. Consider the commutative diagram of homomorphisms of $\bfield[z]$--modules 
\[
\xymatrixcolsep{4pc}\xymatrixrowsep{3pc}\xymatrix{\bfield[z]^{tn} \ar^-{\cdot M(G)}[r] & \bfield[z]^n \ar^-{\pi \mathfrak{p}}[r] \ar^-{\cdot Q}[d] \ar^-{h}[dr] & R/I \ar^-{\gamma}[d] \ar[r] & 0\\
\bfield[z]^{tn} \ar^-{\cdot P}[u] \ar^-{\cdot H}[r] & \bfield[z]^n \ar^-{\cdot V}[r] & \ar@/^/[l]^{\cdot V^T} \ar@/^/[lu]^{s} \bfield[z]^{n-k} \ar[r] & 0}
\]
with exact rows,
where $h$ is defined by the matrix $M_h = QV$, $s$ by the matrix $M_s = V^TQ^{-1}$, and $\gamma$ is uniquely determined by $h$, since $M(G)M_h = M(G)QV = P^{-1}HV = 0$. Moreover, $\gamma$ is an isomorphism because both $\cdot P$ and $\cdot Q$ are isomorphisms, and $s$ splits the epimorphism $h$. The last claim follows because $hs$ is given by the matrix $M_sM_h = V^TQ^{-1}QV = V^TV = I_{n-k}$. Define $\iota : R/I \to R$ by $\iota = \mathfrak{p} s \gamma$. Then
\[
\gamma \pi \iota = \gamma \pi \mathfrak{p} s \gamma = hs\gamma = \gamma,
\]
and $\pi \iota = id_{R/I}$, since $\gamma$ is an isomorphism. According to Corollary \ref{cor:SRsplit}, the homomorphism of left $R$--modules $\beta : R/I \to R$ defined as
\[
\beta (r + I) = \sum_i a_i \iota (b_ir + I)
\]
for all $r + I \in R/I$ splits $\pi$. In particular, $f = \beta(1+I)$ is an idempotent in $R$ which generates a complement of \(I\) in \(R\) and, since $\pi (1- f) = 0$, then $e = 1-f$ generates the left ideal $I$. Now, $f$, and, therefore, $e$, can be explicitly computed:
\[
\begin{split}
f &= \beta(1+I) \\
& = \sum_ia_i\mathfrak{p}s\gamma(b_i + I) \\
& = \sum_ia_i\mathfrak{p}s\gamma(\pi (\mathfrak{p} \mathfrak{v}(b_i) ) \\
& = \sum_ia_i\mathfrak{p}sh(\mathfrak{v}(b_i) ) \\
& = \sum_ia_i\mathfrak{p}(\mathfrak{v}(b_i) M_hM_s). 
\end{split}
\]
The above reasoning proves the correctness of the following algorithm.

\begin{algorithm}[H]
\caption{Computation of the generating idempotent}\label{idem}
\begin{algorithmic}[1]
\REQUIRE $G = \{g_0, \dots, g_{t-1}\} \subseteq R$ non-zero. \textbf{Assumption}. A separability element $\sum_i a_i \otimes_{\bfield[z]} b_i$ is provided for the ring extension $\bfield[z] \subseteq R$.
\ENSURE An idempotent $e\in R$ such that $Re=Rg_0 + \dots + Rg_{t-1}$, or zero if it does not exist.
\STATE{Compute the matrix $M(G)$}
\STATE{Compute the Smith form decomposition $H = PM(G)Q$}
\IF{$H$ is not basic}
\RETURN 0
\ENDIF
\STATE{$V\gets \left(\begin{smallmatrix} 0 \\ I_{n-k} \end{smallmatrix}\right)$, where $k=\mathrm{rank}(H)$} and $n=\mathrm{dim}_{\bfield}(A)$
\STATE{$M_h\gets QV$, $M_s\gets V^TQ^{-1}$, $M\gets M_hM_s$}
\STATE{Compute $f_i=\mathfrak{p}(\mathfrak{v}(b_i)\cdot M)$ for all $i$}
\STATE $f \gets \sum_i a_if_i$
\RETURN $1-f$
\end{algorithmic}
\end{algorithm}

The generating idempotent has the following application. If \(I \leq R\) is an ideal code generated by the idempotent \(e \in I\), then \(r \in I\) if and only if \(r(1-e) = 0\), i.e. \(I\) is the kernel of the morphism of left \(R\)--modules defined as right multiplication by \(1-e\). The matrix construction in \eqref{circulantmatrix} for \(G = \{1-e\}\) provides a parity check matrix for the code. We illustrate this construction in the following example.

\begin{example}\label{commutativebase}
Let us consider the finite field $\bfield{} =\field[4]=\field[2](\alpha)$ and $A=\bfield[x]/(x^5-1)$. Hence, since $x^5-1$ decomposes as the product $(x + 1) \cdot (x^{2} +\alpha x + 1) \cdot (x^{2} +\alpha^2 x + 1)$ in $\bfield[x]$, $A\cong K_0 \times K_1 \times K_2$, where
\begin{gather*}
K_0=\frac{\bfield[x]}{(x+1)}, 
K_1=\frac{\bfield[x]}{(x^{2} +\alpha x + 1)} \text{ and } 
K_2=\frac{\bfield[x]}{(x^2 + \alpha^2 x + 1)}.
\end{gather*}
Following \cite{GluesingSchmale:2004}, in order to find a non-block CCC, we need an automorphism $\sigma:A\to A$ which moves some of the isomorphic copies of the block subfields of $A$. In this case, we consider the isomorphisms 
\[
\begin{split}
\psi: K_1 &\longrightarrow K_2 \\
x &\longmapsto \psi(x) = \alpha^2 x+1 
\end{split}
\qquad
\begin{split}
\psi^{-1}: K_2 &\longrightarrow K_1 \\
x &\longmapsto \psi^{-1}(x) = \alpha x + \alpha
\end{split}
\]
Let $\tau:K_i\to K_i$ be the Frobenius automorphism, i.e, $\tau(a)=a^{4}$ for any $a\in K_i$, for $i=1,2$. Hence, we may consider the automorphism $\sigma:A\to A$ defined by 
\[
\sigma(x) \equiv \sigma(1,x,x) 
= (1,\psi^{-1}(x)^4,\psi(x)^4) 
\equiv x^4 + \alpha^2 x^3 + \alpha x^2 + x,
\]
by using Chinese Remainder Theorem. In order to calculate a separability element $\overline{p}$ of the extension $\field[4][z] \subseteq A[z;\sigma]$ we follow the procedure explained in the proof of Proposition \ref{Aconm}.
We have
\begin{itemize}
\item $\{1\}$ is a self-dual normal basis of $K_0$. 
\item $\{x,x^4\}$ and $\{\alpha x,(\alpha x)^4\}$ are normal dual bases for \(K_1\). 
\item Applying $\psi$, $\{\alpha^2 x+1,(\alpha^2 x+1)^4\}$ and $\{x+\alpha,(x+\alpha)^4\}$ are normal dual bases for $K_2$. 
\end{itemize}
By using Chinese Remainder Theorem, it is straightforward to calculate all these elements in $A$ and compute a separability element $\overline{p}$ according to \eqref{sepele2}:
\[
\begin{split}
\overline{p} &= (x^4 + x^3 + x^2 + x + 1) \otimes_{\field[4][z]} (x^4 + x^3 + x^2 + x + 1) \\
&\quad + (\alpha^2 x^4 + \alpha^2 x^3 + \alpha x^2 + \alpha) \otimes_{\field[4][z]} (x^4 + x^3 + \alpha^2 x^2 + \alpha^2) \\
&\quad + (\alpha x^3 + \alpha^2 x^2 + \alpha^2 x + \alpha) \otimes_{\field[4][z]} (\alpha^2 x^3 + x^2 + x + \alpha^2) \\
&\quad + (\alpha^2 x^4 + \alpha^2 x^2 + \alpha x + \alpha) \otimes_{\field[4][z]} (x^4 + x^2 + \alpha^2 x + \alpha^2) \\
&\quad + (\alpha x^4 + \alpha^2 x^3 + \alpha^2 x + \alpha) \otimes_{\field[4][z]} (\alpha^2 x^4 + x^3 + x + \alpha^2)
\end{split}
\]

Now let \(I\) be the left ideal generated by the Ore polynomial 
\[
g=z^2 (\alpha^2 x^{4} + \alpha x^{3} + \alpha x^{2} + \alpha^2 x) + z (x^{4} + x^3 + x^2 + x) + (\alpha^2 x^{4} + \alpha x^{3} + \alpha x^{2} + \alpha^2 x + 1).
\] 
One may compute \(M(g)\), which is called the $\sigma$-circulant matrix of $g$ in \cite{GluesingSchmale:2004},
\begin{footnotesize}
\[
M(g) = 
\left(\begin{matrix}
1 & \alpha^2 z^{2} + z + \alpha^2 & \alpha z^{2} + z + \alpha & \alpha z^{2} + z + \alpha & \alpha^2 z^{2} + z + \alpha^2 \\
\alpha^2 z^{2} + z + \alpha^2 & 1 & \alpha^2 z^{2} + \alpha^2 z + \alpha^2 & \alpha z^{2} + \alpha z + \alpha & \alpha z^{2} + \alpha
\\
\alpha z^{2} + z + \alpha & \alpha^2 z^{2} + \alpha^2 z + \alpha^2 & 1 & \alpha^2 z^{2} + \alpha^2 & \alpha z^{2} + \alpha z + \alpha
\\
\alpha z^{2} + z + \alpha & \alpha z^{2} + \alpha z + \alpha & \alpha^2 z^{2} + \alpha^2 & 1 & \alpha^2 z^{2} + \alpha^2 z + \alpha^2
\\
\alpha^2 z^{2} + z + \alpha^2 & \alpha z^{2} + \alpha & \alpha z^{2} + \alpha z + \alpha & \alpha^2 z^{2} + \alpha^2 z + \alpha^2 & 1
\end{matrix}\right),
\]
\end{footnotesize}
whose Smith form decomposition is $H = PM(g)Q$, where $H=\left(\begin{smallmatrix} I_3 & 0 \\ 0 & 0 \end{smallmatrix}\right)$ and 
\begin{footnotesize}
\[
Q = \left(\begin{matrix}
1 & \alpha^2 z^{2} + z + \alpha^2 & \alpha z^{3} + \alpha z^{2} + z + \alpha & z^{2} + \alpha^2 z + \alpha & \alpha^2 z^{3} + z^{2} + \alpha z \\
0 & 1 & \alpha z + \alpha^2 & \alpha z^{2} + 1 & z^{3} + z^{2} + z + 1 \\
0 & 0 & z + \alpha & \alpha z^{2} + z + 1 & z^{3} + \alpha z^{2} + \alpha z + \alpha^2 \\
0 & 0 & \alpha^2 z + \alpha & z^{2} + \alpha z + \alpha & \alpha^2 z^{3} + \alpha z^{2} + z + \alpha^2 \\
0 & 0 & 0 & 0 & 1
\end{matrix}\right).
\]
\end{footnotesize}

Therefore, $I$ is a $\sigma$-cycic convolutional code of dimension 3 and length 5. Following Algorithm \ref{idem}, the morphism $h$ and its section $s$ are given by the matrices
\begin{footnotesize}
\[
M_h = \left(\begin{matrix}
z^{2} + \alpha^2 z + \alpha & \alpha^2 z^{3} + z^{2} + \alpha z \\
\alpha z^{2} + 1 & z^{3} + z^{2} + z + 1 \\
\alpha z^{2} + z + 1 & z^{3} + \alpha z^{2} + \alpha z + \alpha^2 \\
z^{2} + \alpha z + \alpha & \alpha^2 z^{3} + \alpha z^{2} + z + \alpha^2 \\
0 & 1
\end{matrix}\right) \text{ and }
M_s = \left(\begin{matrix}
0 & 0 & \alpha^2 z + \alpha & z + \alpha & 0 \\
0 & 0 & 0 & 0 & 1
\end{matrix}\right)
\] 
\end{footnotesize}
Hence, $f_i=\mathfrak{p}(\mathfrak{v}(b_i)\cdot M)$, where
\begin{footnotesize}
\[
M = \left(\begin{matrix}
0 & 0 & \alpha^2 z^{3} + \alpha^2 & z^{3} + z^{2} + \alpha^2 z + \alpha^2 & \alpha^2 z^{3} + z^{2} + \alpha z \\
0 & 0 & z^{3} + \alpha^2 z^{2} + \alpha^2 z + \alpha & \alpha z^{3} + \alpha^2 z^{2} + z + \alpha & z^{3} + z^{2} + z + 1 \\
0 & 0 & z^{3} + z + \alpha & \alpha z^{3} + \alpha z^{2} + \alpha^2 z + \alpha & z^{3} + \alpha z^{2} + \alpha z + \alpha^2 \\
0 & 0 & \alpha^2 z^{3} + \alpha^2 z^{2} + \alpha z + \alpha^2 & z^{3} + z + \alpha^2 & \alpha^2 z^{3} + \alpha z^{2} + z + \alpha^2 \\
0 & 0 & 0 & 0 & 1
\end{matrix}\right)
\]
\end{footnotesize}
The parity check idempotent polynomial is $f=\sum_i a_if_i$. Concretely,
\[
\begin{split}
f &= z^3 (\alpha^2 x^{4} + \alpha^2 x^{3} + x^{2} + 1) + z^2 (\alpha x^{4} + x^{2} + \alpha^2 x) \\
&\quad + z (\alpha x^{4} + x^{3} + \alpha x^{2} + 1) + (\alpha^2 x^{4} + \alpha x^{3} + \alpha x^{2} + \alpha^2 x).
\end{split}
\]
and the output of Algorithm \ref{idem} is 
\[
\begin{split}
e &= z^3 (\alpha^2 x^{4} + \alpha^2 x^{3} + x^{2} + 1) + z^2 (\alpha x^{4} + x^{2} + \alpha^2 x) \\
&\quad + z (\alpha x^{4} + x^{3} + \alpha x^{2} + 1) + (\alpha^2 x^{4} + \alpha x^{3} + \alpha x^{2} + \alpha^2 x + 1).
\end{split}
\]
From the parity check polynomial \(f\), we can compute hence a parity check matrix,
\begin{scriptsize}
\begin{multline*}
M(f) = \\
\left(\begin{matrix}
z^{3} + z & \alpha^2 z^{2} + \alpha^2 & z^{3} +
z^{2} + \alpha z + \alpha & \alpha^2 z^{3} + z +
\alpha & \alpha^2 z^{3} + \alpha z^{2} + \alpha z +
\alpha^2 \\
z^{3} + \alpha z^{2} + z + \alpha^2 & \alpha^2
z^{3} + \alpha^2 z & \alpha^2 z^{3} +
\alpha^2 z^{2} + \alpha z + \alpha^2 & z^{3} +
z^{2} + \alpha^2 z + \alpha & \alpha^2
z + \alpha \\
\alpha^2 z^{3} + \alpha & z^{3} + \alpha z^{2} +
\alpha^2 & 0 & z^{3} + \alpha^2 z^{2} +
\alpha^2 z + \alpha^2 & \alpha^2
z^{3} + z^{2} + \alpha^2 z + \alpha \\
z^{2} + \alpha z + \alpha & z^{3} + \alpha^2 z +
\alpha & \alpha^2 z^{3} + \alpha z^{2} + z + \alpha +
1 & \alpha^2 z^{3} + \alpha^2 z &
z^{3} + \alpha^2 z^{2} + \alpha^2 z +
\alpha^2 \\
\alpha^2 z^{3} + \alpha^2 z^{2} + \alpha z
+ \alpha^2 & \alpha^2 z^{3} + z^{2} + \alpha &
z^{3} + z + \alpha & \alpha z^{2} + \alpha z + \alpha^2 &
z^{3} + z
\end{matrix}\right).
\end{multline*}
\end{scriptsize}

Following the techniques developed in \cite{Gluesing_MDS}, the degree of this code is \(\delta = 2\), hence it is a \((5,3,2)_4\) convolutional code. Then its free distance is bounded by \(5\). Actually, we may calculate the first terms of the column distances of $I$. Concretely, $d^c_0=3$, $d^c_1=4$, $d^c_2=5$. So, the free distance of $I$, $d_{\text{free}} (I)=5$ and it is an MDS code.
\end{example}

\begin{example}[Continuation of Example \ref{2x2mat}]\label{2x2matidemp}
Let \(\mathcal{B} = \left\{ \left(\begin{smallmatrix}
1 & 0 \\
0 & 0
\end{smallmatrix}\right), 
\left(\begin{smallmatrix}
0 & 1 \\
0 & 0
\end{smallmatrix}\right), 
\left(\begin{smallmatrix}
0 & 0 \\
1 & 0
\end{smallmatrix}\right), 
\left(\begin{smallmatrix}
0 & 0 \\
0 & 1
\end{smallmatrix}\right)\right\}\) be the chosen basis of \(\mathcal{M}_2(\field[8])[z;\sigma]\) as \(\field[8][z]\)--module. Let $I$ be the left ideal of $R=\mathcal{M}_2(\field[8])[z;\sigma]$ generated by $g$, where
\[
g=z^2\left(\begin{matrix}
\alpha^5 & \alpha^{6} \\
0 & 0
\end{matrix}\right) + z \left(\begin{matrix}
\alpha^{5} & \alpha^{4} \\
\alpha & 0
\end{matrix}\right)+ \left(\begin{matrix}
1 & 0 \\
\alpha^{6} & 0
\end{matrix}\right).
\]
Hence, 
\[
M(g)=\left(\begin{matrix}
\alpha^{6} z^{2} + \alpha^{5}z + 1 & z^{2} + \alpha^{5} z &
\alpha^{5} z^{2} + \alpha z &
\alpha^{6} z^{2} + \alpha z \\
\alpha^{2} z^{2} + \alpha^{6} & \alpha^{3} z^{2} +
\alpha^{4} z & \alpha z^{2} & \alpha^{2}
z^{2} + z \\
\alpha^{5} z^{2} + \alpha z &
\alpha^{6} z^{2} + \alpha z & \alpha^{2} z^{2} +
\alpha^{2} z + 1 & \alpha^{3} z^{2} + \alpha^{2} z \\
\alpha z^{2} & \alpha^{2} z^{2} + z & \alpha^{5} z^{2} + \alpha^{6} & \alpha^{6}
z^{2} + \alpha z
\end{matrix}\right)
\]
whose Smith form decomposition is $H=PM_gQ$, where
\[
H=\left(\begin{matrix}
1 & 0 & 0 & 0 \\
0 & 1 & 0 & 0 \\
0 & 0 & 0 & 0 \\
0 & 0 & 0 & 0
\end{matrix}\right),
Q=\left(\begin{matrix}
1 & \alpha^5 z & \alpha^4 z^{2} + \alpha^4 z & \alpha^{2} z^{3} + z^{2} + \alpha z \\
0 & \alpha^4 z + \alpha^{2} & \alpha^3 z^{2} + \alpha z + \alpha^{2} + 1 & \alpha z^{3} + \alpha z^{2} + \alpha^3 z \\
0 & \alpha & z & \alpha^5 z^{2} + \alpha^{2} z \\
0 & 0 & 0 & 1
\end{matrix}\right),
\]
\[
\text{and }
P=\left(\begin{matrix}
\alpha^{6} z + \alpha^{2} & \alpha^{3}
z + 1 & 0 & 0 \\
\alpha^{6} z &\alpha^{3} z &
\alpha^{6} & 0 \\
z^{2} + 1 & \alpha^{4} z^{2} +
\alpha^{6} z + \alpha & \alpha z & 0 \\
\alpha^{4} z^{3} + z^{2} + \alpha^{6} z & \alpha z^{3} + \alpha^{6} z^{2} +
\alpha^{3} z & \alpha^{5}
z^{2} + \alpha z + \alpha^{6} & 1
\end{matrix}\right).
\]
Therefore, $I$ is an ideal code of dimension 2 and length 4. Following Algoritm \ref{idem}, the morphism $h$ and its section $s$ are given by the matrices
\[
M_h=\left(\begin{matrix}
\alpha^{4} z^{2} + \alpha^{4} z & \alpha^{2} z^{3} + z^{2} + \alpha z \\
\alpha^{3} z^{2} + \alpha z + \alpha^{6} & \alpha z^{3} + \alpha z^{2} + \alpha^{3} z \\
z & \alpha^{5} z^{2} + \alpha^{2} z \\
0 & 1
\end{matrix}\right) \text{ and } 
M_s=\left(\begin{matrix}
0 & \alpha & \alpha^{4} z + \alpha^{2} & 0 \\
0 & 0 & 0 & 1
\end{matrix}\right).
\] 
Hence, $f_i = \mathfrak{p}(\mathfrak{v}(b_i)\cdot M)$, where
\[
M=\left(\begin{matrix}
0 & \alpha^{5} z^{2} + \alpha^{5} z & \alpha z^{3} + \alpha^{5} z^{2} +
 \alpha^{6} z & \alpha^{2} z^{3} + z^{2} + \alpha z \\
0 & \alpha^{4} z^{2} + \alpha^{2} z + 1 &
z^{3} + \alpha & \alpha z^{3} + \alpha z^{2} + \alpha^{3} z \\
0 & \alpha z & \alpha^{4} z^{2} +
\alpha^{2} z & \alpha^{5} z^{2} +
\alpha^{2} z \\
0 & 0 & 0 & 1
\end{matrix}\right).
\]
Now, the parity check idempotent polynomial is $f =\sum_ia_if_i$. Concretely,
\[
f=z^3\left(\begin{matrix}
\alpha^{6} & 1 \\
\alpha^{5} & \alpha^{6}
\end{matrix}\right) + z^2 \left(\begin{matrix}
\alpha^{3} & \alpha^{2} \\
\alpha^{2} & \alpha^{6}
\end{matrix}\right) + z\left(\begin{matrix}
\alpha^{4} & \alpha^{4} \\
1 & 0
\end{matrix}\right) + \left(\begin{matrix}
0 & 0 \\
\alpha^{6} & 1
\end{matrix}\right)
\]
and the generating idempotent of $I$ is
\[
e=1-f=z^3\left(\begin{matrix}
\alpha^{6} & 1 \\
\alpha^{5} & \alpha^{6}
\end{matrix}\right) + z^2 \left(\begin{matrix}
\alpha^{3} & \alpha^{2} \\
\alpha^{2} & \alpha^{6}
\end{matrix}\right) + z\left(\begin{matrix}
\alpha^{4} & \alpha^{4} \\
1 & 0
\end{matrix}\right) + \left(\begin{matrix}
1 & 0 \\
\alpha^{6} & 0
\end{matrix}\right).
\]
Again, we may calculate the first terms of the column and row distances of the ideal code $I$. Concretely, $d^c_0=1$, $d^c_1=3$, $d^c_2=4$ and $d^r_0=4$. Hence, by \cite{Gluesing_MDS}, the free distance of $I$ is $d_{\text{free}} (I)=4$.
\end{example}


 \end{document}